\documentclass[11pt,oneside,english]{amsart}
\usepackage[T1]{fontenc}
\usepackage[latin9]{inputenc}
\usepackage{babel}
\usepackage{amsthm}
\usepackage{amssymb}
\usepackage{esint}
\usepackage[unicode=true,pdfusetitle,
 bookmarks=true,bookmarksnumbered=false,bookmarksopen=false,
 breaklinks=false,pdfborder={0 0 1},backref=false,colorlinks=false]
 {hyperref}
\usepackage{breakurl}

\makeatletter
\numberwithin{equation}{section}
\numberwithin{figure}{section}
\theoremstyle{plain}
\newtheorem{thm}{\protect\theoremname}[section]
  \theoremstyle{definition}
  \newtheorem{defn}[thm]{\protect\definitionname}
  \theoremstyle{plain}
  \newtheorem{cor}[thm]{\protect\corollaryname}
  \theoremstyle{remark}
  \newtheorem{rem}[thm]{\protect\remarkname}
  \theoremstyle{plain}
  \newtheorem{lem}[thm]{\protect\lemmaname}
  \theoremstyle{plain}
  \newtheorem{prop}[thm]{\protect\propositionname}

\usepackage{ifpdf} 
\ifpdf 

 \IfFileExists{lmodern.sty}{\usepackage{lmodern}}{}

\fi 

\usepackage[figure]{hypcap}

\makeatother

  \providecommand{\corollaryname}{Corollary}
  \providecommand{\definitionname}{Definition}
  \providecommand{\lemmaname}{Lemma}
  \providecommand{\propositionname}{Proposition}
  \providecommand{\remarkname}{Remark}
\providecommand{\theoremname}{Theorem}

\begin{document}

\title{Uniform Distribution of Eigenstates on a Torus with Two Point Scatterers}

\author{Nadav Yesha}

\date{\today}
\begin{abstract}
We study the Laplacian perturbed by two delta potentials on a two-dimensional
flat torus. There are two types of eigenfunctions for this operator:
old, or unperturbed eigenfunctions which are eigenfunctions of the
standard Laplacian, and new, perturbed eigenfunctions which are affected
by the scatterers. We prove that along a density one sequence, the
new eigenfunctions are uniformly distributed in configuration space,
provided that the difference of the scattering points is Diophantine. 
\end{abstract}

\address{School of Mathematics, University of Bristol, Bristol BS8 1TW, UK}

\email{nadav.yesha@bristol.ac.uk}



\maketitle

\section{Introduction}

\markright{UNIFORM DISTRIBUTION ON A TORUS WITH TWO POINT SCATTERERS}

\subsection{Toral Point Scatterers}

In the field of Quantum Chaos, one of the fundamental questions concerns
the quantum ergodicity of a quantum system, i.e., equidistribution
of almost all eigenstates of the system in the high energy limit.
A key result is Shnirelman's quantum ergodicity theorem \cite{Shnirelman,CdV2,Zelditch},
which asserts that a quantum system whose classical counterpart has
ergodic dynamics is quantum ergodic. On the other hand, there are
quantum systems whose classical counterpart has integrable dynamics,
for which the eigenstates tend to localize (``scar'') in phase space.

A point scatterer on a flat torus is a popular model to study the
transition between integrable and chaotic systems. Formally, it is
defined as a rank one perturbation of the Laplacian, namely
\begin{equation}
-\Delta+\alpha \left<\delta_{x_{0}}, \cdot\right>\delta_{x_{0}} \label{eq:OneScatterer}
\end{equation}
where $\alpha \in \mathbb{R}$ is a coupling parameter, and $x_{0}$ is the scattering
point. It is an intermediate model, in the sense that the delta potential
at $x_{0}$ does not change the (integrable) classical dynamics of
the system except for a measure zero set of trajectories, whereas
it has a chaotic influence on the behaviour of the quantum system.

A standard way to rigorously define a point scatterer is via the theory
of self-adjoint extensions, as described in depth in \cite{CdV1}.
One defines the operator (\ref{eq:OneScatterer}) as a self-adjoint
extension of the Laplacian vanishing near the point $x_{0}$ (there
are non-trivial self-adjoint extensions only in dimensions $d\le3$);
such extensions are parametrized by a phase $\phi\in(-\pi,\pi]$,
where $\phi=\pi$ corresponds to $\alpha=0$ in (\ref{eq:OneScatterer}),
i.e., to the standard Laplacian. Consider the other, non-trivial extensions.
Their eigenfunctions can be split into eigenfunctions of the standard
Laplacian, referred to as the old, or unperturbed eigenfunctions,
as well as new, or perturbed eigenfunctions which are affected by
the presence of the scatterer, and therefore are the main object of
study. 

The semiclassical limits for the perturbed eigenfunctions of a point
scatterer on flat tori have been studied extensively in recent
years (see \cite{Ueberschaer} for a survey on some of the results).
Rudnick and Ueberschär proved uniform distribution in configuration
space of the perturbed eigenfunctions for a point scatterer on two-dimensional
flat tori \cite{RudnickUeberschaer}. This was also proved for three-dimensional
flat tori \cite{Yesha}, both on the standard square torus and on
irrational tori with a Diophantine condition on the side lengths,
where in the former case of the standard torus \emph{all} of the perturbed
eigenfunctions equidistribute in configuration space. As for quantum
ergodicity in full phase space, it was proved both on the standard
two-dimensional flat torus \cite{KurlbergUeberschaer} and on the
standard three-dimensional torus \cite{Yesha2}. 

Scarring behavior has also been studied in several settings. Kurlberg
and Ueberschär showed \cite{KurlbergUeberschaer2} that for an irrational
two-dimensional torus (also known as the ``Šeba billiard'' as introduced
in \cite{Seba}) with a Diophantine condition on the side lengths,
quantum ergodicity does not hold in full phase space; in fact, almost
all new eigenfunctions strongly localize in momentum space. More recently,
Kurlberg and Rosenzweig studied scarring behaviour on standard tori
both in two and three dimensions \cite{KurlbergRosenzweig}.

\subsection{Two Point Scatterers}

Recently, Ueberschär raised the natural question of the behavior of
a system with several scatterers \cite{Ueberschaer2,Ueberschaer3,Ueberschaer4}.
For a standard torus with $n$ i.i.d uniform random scatterers, he
showed \cite{Ueberschaer3} that uniform distribution in configuration
space of almost all of the perturbed eigenfunctions holds with probability
$\gg1/n$. Our goal in this paper is to prove a deterministic result
for two point scatterers on the torus.

Interestingly, our techniques do not generalize to the case of three
or more scatterers, for which the symmetries that we exploit fail
to hold. Indeed, it seems that even one additional (third) scatterer
significantly complicates the nature of the system, so deterministic
results for three or more scatterers require additional arguments.
As an example, note that even with the presence of only a few scatterers, some unique phenomena occur, such as Laplace eigenspaces
of dimensions smaller than the number of scatterers (this can occur
for all eigenspaces, e.g. for irrational tori). Additional arguments are also required in order to extend our equidistribution results to phase space (as in \cite{KurlbergUeberschaer}) or to study scarring behaviour for systems with multiple scatterers. 

For the clarity of the paper, we will not work in the most general
setting. Here we consider the two-dimensional standard flat torus
$\mathbb{T}^{2}=\mathbb{R}^{2}/2\pi\mathbb{Z}^{2}$ with two scatterers
at the points $x_{1},x_{2}\in\mathbb{T}^{2}$, whose normalized difference
$\left(x_{2}-x_{1}\right)/\pi$ is Diophantine. Our results can be
easily generalized to non-square tori with a Diophantine condition
on the difference of the scatterers -- see Theorem \ref{thm:NonSquareTorus}
below. In addition, using the methods of \cite{Yesha}, Theorem \ref{thm:MainResultGeneral}
can be extended to the standard three-dimensional torus, and also to irrational tori with
the same Diophantine condition on the side lengths as in \cite{Yesha}.

To give a more detailed account of our results, recall the definition
of a Diophantine vector:
\begin{defn}
A vector $\left(\alpha_{1},\alpha_{2}\right)\in\mathbb{R}^{2}$ is
said to be Diophantine of type $\kappa$, if there exists a constant
$C>0$ such that
\[
\max_{j=1,2}\left|\alpha_{j}-\frac{p_{j}}{q}\right|>\frac{C}{q^{\kappa}}
\]
for all $p_{1},p_{2},q\in\mathbb{Z},$ $q>0$. By Dirichlet's theorem,
the smallest possible value for $\kappa$ is $3/2$.
\end{defn}
Let $x_{0}=x_{2}-x_{1}$, and assume that the vector $x_{0}/\pi$
is Diophantine. Note that by Khinchin's theorem on Diophantine approximations,
our assumption holds for almost all pairs $x_{1},x_{2}$.

Consider the Laplacian perturbed by two delta potentials at $x_{1},x_{2}$, which is formally defined as a symmetric rank two perturbation of the Laplacian, namely
\begin{equation}
-\Delta + \sum_{i,j=1}^{2}h_{ji}\left< \delta_{x_{i}}, \cdot \right>\delta_{x_{j}} \label{eq:TwoScatterers}
\end{equation}
where $H = \left\{h_{ij}\right\}_{i,j=1}^2$ is a Hermitian matrix, and $x_1, x_2 \in \mathbb{T}^2$ are the scattering points (in particular, the case where we assume no non-local interaction between the scatterers corresponds to perturbations with a diagonal $ H $, i.e.,
\[  -\Delta + \alpha\left< \delta_{x_{1}}, \cdot \right>\delta_{x_{1}} + 
\beta\left< \delta_{x_{2}}, \cdot \right>\delta_{x_{2}}
\] where $ \alpha,\beta\in \mathbb{R} $ are coupling constants).
As in the case of a single scatterer, the formal operators  (\ref{eq:TwoScatterers}) can be defined
rigorously using the theory of self-adjoint extensions of the standard
Laplacian vanishing at $x_{1},x_{2}$. The self-adjoint extensions
are parametrized by the unitary group $U\left(2\right)$. As we will
see, the standard Laplacian is retrieved by the extension corresponding
to the matrix $U=-I$, and the eigenvalues of the other extensions
can be again divided into old or unperturbed eigenvalues, i.e., eigenvalues
of the standard Laplacian, as well as a set of new, perturbed eigenvalues,
which we denote by $\Lambda=\Lambda_{U}$.

To establish a link between the old and the new eigenvalues of a self-adjoint
extension $-\Delta_{U}$, we define the ``weak interlacing'' property:
\begin{defn}
We say that a set $A\subseteq\mathbb{R}$ weakly interlaces with a
set $B\subseteq\mathbb{R}$, if $ A \cap B = \emptyset $, and there exists a constant $C>0$ such
that between any two elements of $A$ there are at most $C$ elements
of $B$, and vice versa. 
\end{defn}
It is a general fact \cite{BirmanSolomjak} that for $n$ point scatterers
(which are similarly defined via self-adjoint extensions), the difference
between the spectral counting function of $-\Delta_{U}$ (with multiplicities)
and the spectral counting function of the standard Laplacian is uniformly
bounded by $n$. In Appendix \ref{sec:Appendix-A}, we will see that
for each $0\ne\lambda\in\sigma\left(-\Delta\right),$ the dimension
of the corresponding eigenspace of $-\Delta_{U}$ is equal to the
dimension of the Laplace eigenspace minus $\mbox{rank}\left(I+U\right)$.
It follows that the set $\Lambda$ of new eigenvalues weakly interlaces
with the Laplace eigenvalues.

\subsection{Statement of the Main Result}

We now state our main result. Let $\Lambda_{0}$ be any set of real
numbers which weakly interlaces with the Laplace eigenvalues. For
$\lambda\in\Lambda_{0}$ and $\left(d_{1},d_{2}\right)\ne\left(0,0\right)$,
let 
\[
G_{\lambda}\left(x\right)=G_{\lambda}\left(x;d_{1},d_{2}\right)=d_{1}G_{\lambda}\left(x,x_{1}\right)+d_{2}G_{\lambda}\left(x,x_{2}\right)
\]
be any non-zero superposition of the Green's functions 
\[
G_{\lambda}\left(x,x_{j}\right)=\left(\Delta+\lambda\right)^{-1}\delta_{x_{j}}\left(x\right)\hspace{1em}j=1,2,
\]
and let $g_{\lambda}=G_{\lambda}/\left\Vert G_{\lambda}\right\Vert _{2}$.
\begin{thm}
\label{thm:MainResultGeneral}Let $x_{0}=x_{2}-x_{1},$ and assume
that $x_{0}/\pi$ is Diophantine. For any $ \epsilon>0 $ and for any set $\Lambda_{0}\subseteq\mathbb{R}$
which weakly interlaces with the Laplace eigenvalues, there exists
a subset $\Lambda_{\infty}\subseteq\Lambda_{0}$ of density one so
that for all observables $a\in C^{\infty}\left(\mathbb{T}^{2}\right)$,
\[
\int_{\mathbb{T}^{2}}a\left(x\right)\left|g_{\lambda}\left(x;d_{1},d_{2}\right)\right|^{2}\,\mbox{d}x=\frac{1}{4\pi^{2}}\int_{\mathbb{T}^{2}}a\left(x\right)\,\mbox{d}x+O_{a,\epsilon}\left(\lambda^{-1/8+\epsilon}\right)
\]
as $\lambda\to\infty$ along $\Lambda_{\infty}$.
\end{thm}
Let $-I\ne U\in U\left(2\right)$, and let $-\Delta_{U}$ the corresponding
self-adjoint extension. For a new eigenvalue $\lambda\in\Lambda$,
the corresponding eigenfunction is a superposition of the Green's
functions $G_{\lambda}\left(x,x_{j}\right)$. Thus, given an orthonormal
basis $\left\{ \varphi_{k}\right\} $ for the subspace of the perturbed
eigenfunctions, it follows from Theorem \ref{thm:MainResultGeneral}
and from the weak interlacing of $\Lambda$ with the eigenvalues of
$-\Delta$ that $\left\{ \varphi_{k}\right\} $ is uniformly distributed
in configuration space along a density one subsequence:
\begin{cor}
Let $x_{0}=x_{2}-x_{1},$ and assume that $x_{0}/\pi$ is Diophantine. For any $-I\ne U\in U\left(2\right)$, let $\left\{ \varphi_{k}\right\} $
be an orthonormal basis for the subspace of the perturbed eigenfunctions
of $-\Delta_{U}$ with eigenvalues $\left\{ \lambda_{k}\right\} $.
For any $ \epsilon>0 $, there exists a density one sequence $\left\{ \lambda_{k_j}\right\} $ so that
for all observables $a\in C^{\infty}\left(\mathbb{T}^{2}\right)$,
\[
\int_{\mathbb{T}^{2}}a\left(x\right)\left|\varphi_{k_{j}}\left(x\right)\right|^{2}\,\mbox{d}x=\frac{1}{4\pi^{2}}\int_{\mathbb{T}^{2}}a\left(x\right)\,\mbox{d}x+O_{a,\epsilon}\left(\lambda_{k_{j}}^{-1/8+\epsilon}\right)
\]
as $j\to\infty$.
\end{cor}
In particular, we improve on the result of Ueberschär \cite[Theorem 1.1]{Ueberschaer3}
for two scatterers, as in that case his result only gives the result
for random $x_{1},x_{2}$ in a set of positive, but not necessarily
full measure. Our result is deterministic and applies for almost all
$x_{1},x_{2}$. 

Note that the formulation of Theorem \ref{thm:MainResultGeneral}
is fairly general, and is independent of the self-adjoint extension
$U$, which is advantageous since in the physics literature one often
considers self-adjoint extensions which are not fixed but vary with
$\lambda$. For a single scatterer, for example, there is a popular
quantization condition known as the ``strong coupling limit'' where
$\tan\frac{\phi}{2}\sim-C\log\lambda$ (see \cite{Shigehara,Ueberschaer}),
in which phenomena such as level repulsion between the new eigenvalues
are observed, as opposed to the ``weak coupling limit'' where the
self-adjoint extension is fixed. In particular, it follows from Theorem
\ref{thm:MainResultGeneral} that uniform distribution in configuration
space holds even if the self-adjoint extensions change with $\lambda.$ 

As stated above, the proof of Theorem \ref{thm:MainResultGeneral}
is easily generalized to non-square tori: For $a>0$, define a lattice
$\mathcal{L}_{0}=\mathbb{Z}\left(1/a,0\right)\oplus\mathbb{Z}\left(0,a\right)$
in $\mathbb{R}^{2},$ and let 
\[
\mathcal{L}=\left\{ x\in\mathbb{R}^{2}:\,\left\langle x,l\right\rangle \in\mathbb{Z},\,\forall l\in\mathcal{L}_{0}\right\} =\mathbb{Z}\left(a,0\right)\oplus\mathbb{Z}\left(0,1/a\right)
\]
be the dual lattice. Consider the torus $\mathbb{T}_{\mathcal{L}_{0}}^{2}=\mathbb{R}^{2}/2\pi\mathcal{L}_{0}$
with scattering points $x_{1},x_{2}\in\mathbb{T}_{\mathcal{L}_{0}}^{2}$,
whose difference we denote by $x_{2}-x_{1}=\left(\alpha_{1},\alpha_{2}\right)$,
and assume that $\left(\alpha_{1}a/\pi,\alpha_{2}/\left(\pi a\right)\right)$
is Diophantine. Let $\Lambda_{0}$ be any set of real numbers which
weakly interlaces with the Laplace eigenvalues, which are the norms
$\left|\xi\right|^{2}$ of the elements $\xi\in\mathcal{L}$. For
$\lambda\in\Lambda_{0}$ and $\left(d_{1},d_{2}\right)\ne\left(0,0\right)$,
let 
\[
G_{\lambda}\left(x\right)=G_{\lambda}\left(x;d_{1},d_{2}\right)=d_{1}G_{\lambda}\left(x,x_{1}\right)+d_{2}G_{\lambda}\left(x,x_{2}\right)
\]
be any non-zero superposition of the Green's functions $\left(\Delta+\lambda\right)^{-1}\delta_{x_{j}}$
and let $g_{\lambda}=G_{\lambda}/\left\Vert G_{\lambda}\right\Vert _{2}$.
\begin{thm}
\label{thm:NonSquareTorus}Let $x_{2}-x_{1}=\left(\alpha_{1},\alpha_{2}\right),$
and assume that $\left(\alpha_{1}a/\pi,\alpha_{2}/\left(\pi a\right)\right)$
is Diophantine. There exists a constant $\gamma>0$, such that for any $ \epsilon>0 $ and for any set $\Lambda_{0}\subseteq\mathbb{R}$ which weakly interlaces
with the Laplace eigenvalues, there exists a subset $\Lambda_{\infty}\subseteq\Lambda_{0}$
of density one so that for all observables $a\in C^{\infty}\left(\mathbb{T}_{\mathcal{L}_{0}}^{2}\right)$,
\[
\int_{\mathbb{T}_{\mathcal{L}_{0}}^{2}}a\left(x\right)\left|g_{\lambda}\left(x;d_{1},d_{2}\right)\right|^{2}\,\mbox{d}x=\frac{1}{4\pi^{2}}\int_{\mathbb{T}_{\mathcal{L}_{0}}^{2}}a\left(x\right)\,\mbox{d}x+O_{a,\epsilon}\left(\lambda_{k_{j}}^{-\gamma+\epsilon}\right)
\]
as $\lambda\to\infty$ along $\Lambda_{\infty}$. \end{thm}
\begin{rem}
One can take $\gamma=\left(1-3\theta\right)/2$, where $\theta$ is
the exponent in the remainder term in Weyl's law on the torus (see \cite{Ueberschaer3}). The
best known exponent $\theta=131/416$ is due to Huxley \cite{Huxley},
so one can take $\gamma=23/832$.
\end{rem}

\subsection*{Acknowledgements}

The research leading to these results has received funding from the
European Research Council under the European Union's Seventh Framework
Programme (FP/2007-2013) / ERC Grant Agreement n. 291147. I would
like to thank Zeév Rudnick and Jens Marklof for helpful comments.

\section{Two Points Scatterers on the Torus}

\subsection{Self-Adjoint Extensions}

Let $\mathbb{T}^{2}=\mathbb{R}^{2}/2\pi\mathbb{Z}^{2}$ be the standard
two-dimensional flat torus. Let $x_{1},x_{2}\in\mathbb{T}^{2}$ two
points on the torus, and denote the difference of $x_{1}$ and $x_{2}$
by $x_{0}=x_{2}-x_{1}=\left(\alpha_{1},\alpha_{2}\right)$. Recall
that we assume that $x_{0}/\pi$ is Diophantine. We rigorously define
the Laplacian perturbed by potentials at $x_{1},x_{2}$ using the
theory of self-adjoint extensions of unbounded symmetric operators.
We give here a brief summary of the procedure -- a more general calculation
for $n$ scatterers can be found in \cite{Ueberschaer3}, however
in the case of two scatterers we are able to give a more explicit
computation.

Let $D_{0}=C_{c}^{\infty}\left(\mathbb{T}^{2}\setminus\left\{ x_{1},x_{2}\right\} \right)$
be the space of smooth functions supported away from the points $x_{1},x_{2}$,
and let $-\Delta_{0}=-\Delta_{\restriction D_{0}}$ be the Laplacian
restricted to this domain. We realize the perturbed operator as a
self-adjoint extension of the operator $-\Delta_{0}$. In fact, it
can be shown that the deficiency indices of $-\Delta_{0}$ are $\left(2,2\right)$,
hence the self-adjoint extensions are parametrized by the unitary
group $U\left(2\right)$.

For $\lambda\notin\sigma\left(-\Delta\right)$, let
\[
G_{\lambda}\left(x,y\right)=\left(\Delta+\lambda\right)^{-1}\delta_{y}\left(x\right)
\]
be the Green's function of the Laplacian on $\mathbb{T}^{2}$. In
particular, it has the $L^{2}$-expansion
\[
G_{\lambda}\left(x,y\right)=-\frac{1}{4\pi^{2}}\sum_{\xi\in\mathbb{Z}^{2}}\frac{e^{i\left\langle \xi,x-y\right\rangle }}{\left|\xi\right|^{2}-\lambda}.
\]

The deficiency subspaces of $-\Delta_{0}$, namely $\ker\left(\Delta_{0}^{*}\pm i\right)$
are spanned by 
\[
\left\{ G_{i}\left(x,x_{1}\right),G_{i}\left(x,x_{2}\right)\right\} ,\left\{ G_{-i}\left(x,x_{1}\right),G_{-i}\left(x,x_{2}\right)\right\} .
\]

Note that $G_{-i}\left(x,x_{j}\right)=\overline{G_{i}\left(x,x_{j}\right)}$,
and that for $\lambda\in\mathbb{R}$, $G_{\lambda}\left(x,x_{j}\right)$ is real. 

Let
\[
c_{1}=\left\Vert G_{\pm i}\left(x,x_{j}\right)\right\Vert _{2}^{2}=\frac{1}{16\pi^{4}}\sum_{\xi\in\mathbb{Z}^{2}}\frac{1}{\left|\xi\right|^{4}+1},
\]
\begin{alignat*}{1}
c_{2} & =\int G_{i}\left(x,x_{1}\right)\overline{G_{i}\left(x,x_{2}\right)}\,\mbox{d}x\\
 & =\int G_{-i}\left(x,x_{1}\right)\overline{G_{-i}\left(x,x_{2}\right)}\,\mbox{d}x=\frac{1}{16\pi^{4}}\sum_{\xi\in\mathbb{Z}^{2}}\frac{\cos\left(\left\langle \xi,x_{0}\right\rangle \right)}{\left|\xi\right|^{4}+1}.
\end{alignat*}
Thus, defining 
\[
\mathbb{G}_{\lambda}\left(x\right)=\left(G_{\lambda}\left(x,x_{1}\right),G_{\lambda}\left(x,x_{2}\right)\right)
\]
(for notational convenience we treat $\mathbb{G}_{\lambda}\left(x\right)$
as a vector with two coordinates) and
\[
T=\left(\begin{array}{cc}
\frac{1}{\sqrt{c_{1}}} & 0\\
-\frac{c_{2}}{\sqrt{c_1(c_{1}^{2}-c_{2}^{2})}} & \sqrt{\frac{c_{1}}{c_{1}^{2}-c_{2}^{2}}}
\end{array}\right),
\]
we get that $T\mathbb{G}_{i}\left(x\right)$ and $T\mathbb{G}_{-i}\left(x\right)$
form orthonormal bases for the deficiency spaces $\ker\left(\Delta_{0}^{*}\pm i\right).$

Denote the self-adjoint extension of $-\Delta_{0}$ corresponding
to $U\in U\left(2\right)$ by $-\Delta_{U}$. The domain of $-\Delta_{U}$
consists of the functions in the Sobolev space $H^{2}\left(\mathbb{T}^{2}\setminus\left\{ x_{1},x_{2}\right\} \right)$
of the form
\begin{equation}
g\left(x\right)=f\left(x\right)+\left\langle v,T\mathbb{G}_{i}\left(x\right)\right\rangle +\left\langle v,UT\mathbb{G}_{-i}\left(x\right)\right\rangle \label{eq:NewEigenfunctions}
\end{equation}
where $f\in H^{2}\left(\mathbb{T}^{2}\right)$ such that $f\left(x_{1}\right)=f\left(x_{2}\right)=0$
and $v\in\mathbb{C}^{2}$. 

We can also rewrite (\ref{eq:NewEigenfunctions}) as
\[
g\left(x\right)=f\left(x\right)+\left\langle T^{*}\left(I+U^{*}\right)v,\mbox{Re}\mathbb{G}_{i}\left(x\right)\right\rangle +i\left\langle T^{*}\left(I-U^{*}\right)v,\mbox{Im}\mathbb{G}_{i}\left(x\right)\right\rangle .
\]
Since $\mbox{Im}G_{i}\left(x,x_{j}\right)\in H^{2}\left(\mathbb{T}^{2}\right)$,
we see that the extension $-\Delta_{-I}$ retrieves the standard Laplacian
$-\Delta$ on $H^{2}\left(\mathbb{T}^{2}\right)$, and that $g\in H^{2}\left(\mathbb{T}^{2}\right)$
if and only if $\left(I+U^{*}\right)v=0$. 

Another class of special extensions are $-\Delta_{U}$ where $\mbox{rank}\left(I+U\right)=1$.
For these extensions, there is a non-zero $v_{0}$ (unique up to multiplication
by a scalar) such that $\left(I+U^{*}\right)v_{0}=0$, and therefore
for the choice $v=cv_{0}$, (\ref{eq:NewEigenfunctions}) reads 
\[
g\left(x\right)=f\left(x\right)+2ic\left\langle T^{*}v_{0},\mbox{Im}\mathbb{G}_{i}\left(x\right)\right\rangle 
\]
so $g\in H^{2}\left(\mathbb{T}^{2}\right)$. Since
\begin{equation}
\mbox{Im}G_{i}\left(x_{1},x_{1}\right)=\mbox{Im}G_{i}\left(x_{2},x_{2}\right)=-4\pi^{2}c_{1},\label{eq:c_1_def}
\end{equation}
\begin{equation}
\mbox{Im}G_{i}\left(x_{1},x_{2}\right)=\mbox{Im}G_{i}\left(x_{2},x_{1}\right)=-4\pi^{2}c_{2}\label{eq:c_2_def}
\end{equation}
we see that $\left\langle T^{*}v_{0},\mbox{Im}\mathbb{G}_{i}\left(x\right)\right\rangle $
and therefore $g$ do not vanish simultaneously at $x_{1},x_{2}$.
Thus, if $\mbox{rank}\left(I+U\right)=1$, then there exists $g\in\mbox{Dom}\left(-\Delta_{U}\right)$
such that $g\in H^{2}\left(\mathbb{T}^{2}\right)$ with either $g\left(x_{1}\right)\ne0$
or $g\left(x_{2}\right)\ne0,$ a phenomenon which does not occur in
the case of a single scatterer.

\subsection{Spectrum and Eigenfunctions}

The eigenvalues of $-\Delta_{U}$ for $U\ne-I$, and the corresponding
eigenfunctions, fall into two kinds: First, there are the ``old'',
or ``unperturbed'' eigenvalues, which are the eigenvalues of the
standard Laplacian $-\Delta$ on $\mathbb{T}^{2}$, i.e., belong to
the set $\mathcal{N}$ of integers which are representable as a sum
of two squares. For each $0\ne\lambda\in\sigma\left(-\Delta\right),$
we will see in Appendix \ref{sec:Appendix-A} that every eigenfunction
of $-\Delta_{U}$ with an eigenvalue $\lambda$ is also an eigenfunction
of $-\Delta$. From this we will deduce that the dimension of the
corresponding eigenspace of $-\Delta_{U}$ is equal to the dimension
of the Laplace eigenspace minus $\mbox{rank}\left(I+U\right)$.

The second group of eigenvalues of $-\Delta_{U}$ will be referred
to as the group of new, or perturbed eigenvalues. These are the eigenvalues
that are affected by the scatterers, and therefore are the main object
of our study. Denote the set of the perturbed eigenvalues of $-\Delta_{U}$
by $\Lambda=\Lambda_{U}$.

For $\lambda\in\Lambda$, the corresponding eigenfunction $G_{\lambda}$
is of the form
\[
G_{\lambda}\left(x\right)=f\left(x\right)+\left\langle T^{*}v,\mathbb{G}_{i}\left(x\right)\right\rangle +\left\langle \left(UT\right)^{*}v,\mathbb{G}_{-i}\left(x\right)\right\rangle 
\]
where $f\left(x_{1}\right)=f\left(x_{2}\right)=0$ and $v\notin\mbox{Ker}\left(I+U^{*}\right)$.

Since $G_{\lambda}$ is an eigenvalue of $-\Delta_{U}$, it is also
an eigenvalue of the adjoint operator $-\Delta_{0}^{*}.$ In addition,
we have $\Delta_{0}^{*}G_{\pm i}\left(x,x_{j}\right)=\mp iG_{\pm i}\left(x,x_{j}\right),$
so
\begin{equation}
0=\left(\Delta_{0}^{*}+\lambda\right)G_{\lambda}=\left(\Delta+\lambda\right)f+\left(-i+\lambda\right)\left\langle T^{*}v,\mathbb{G}_{i}\right\rangle +\left(i+\lambda\right)\left\langle \left(UT\right)^{*}v,\mathbb{G}_{-i}\right\rangle \label{eq:EigenfunctionEq}
\end{equation}
and after simplifying using the resolvent identity
\[
\frac{\mp i+\lambda}{\left(\Delta+\lambda\right)\left(\Delta\pm i\right)}=\frac{-1}{\Delta+\lambda}+\frac{1}{\Delta\pm i}
\]
we get 
\[
0=f+\left\langle T^{*}v,\mathbb{G}_{i}-\mathbb{G}_{\lambda}\right\rangle +\left\langle \left(UT\right)^{*}v,\mathbb{G}_{-i}-\mathbb{G}_{\lambda}\right\rangle =f+\left\langle v,\mathbb{A}_{\lambda}\right\rangle 
\]
where
\[
\mathbb{A}_{\lambda}\left(x\right)=T\left(\mathbb{G}_{i}-\mathbb{G}_{\lambda}\right)\left(x\right)+UT\left(\mathbb{G}_{-i}-\mathbb{G}_{\lambda}\right)\left(x\right).
\]
Evaluating at $x=x_{1},x_{2}$, we see that a necessary condition
on $\lambda$ being a new eigenvalue is that
\[
\det\left(\mathbb{A}_{\lambda}\left(x_{1}\right),\mathbb{A}_{\lambda}\left(x_{2}\right)\right)=0.
\]
We remark that the condition is also sufficient, since if the determinant
is zero, we can easily construct $G_{\lambda}$. Also note that
\[
G_{\lambda}\left(x\right)=\left\langle T^{*}\left(I+U^{*}\right)v,\mathbb{G}_{\lambda}\left(x\right)\right\rangle ,
\]
so the perturbed eigenfunctions are linear combinations of the Green's
functions $G_{\lambda}\left(x,x_{j}\right).$

\section{Uniform Distribution in Configuration Space}

\subsection{Density One Subsequence}

Let $\Lambda_{0}$ be a set of real numbers which weakly interlaces
with the Laplace eigenvalues. We first build a density one subsequence
in $\Lambda_{0}$ along which we will be able to obtain a lower bound
for the $L^{2}$-norm of $G_{\lambda}$.

Recall that $\mathcal{N}$ is the set of integers representable as
a sum of two squares, i.e., the eigenvalues of $-\Delta$, and let
$r_{2}\left(n\right)$ be the number of such representations. For
any $\lambda\in\Lambda_{0}$, define 
\[
n_{\lambda}=\max\left\{ n\in\mathcal{N}:\,n<\lambda\right\} .
\]
Fix a small $\epsilon>0$.
\begin{lem}
\label{lem:NandLambdaClose}There exists a density one subsequence
$\Lambda_{1}$ in $\Lambda_{0}$, such that for every $\lambda\in\Lambda_{1}$,
we have $\lambda-n_{\lambda}\ll\lambda^{\epsilon}.$\end{lem}
\begin{proof}
Denote $\mathcal{N}=\left\{ n_{1},n_{2},\dots\right\} .$ By Lemma
2.1 in \cite{RudnickUeberschaer}, for a density one sequence $\mathcal{N}_{1}$
in $\mathcal{N}$ we have $n_{k+1}-n_{k}\ll n_{k}^{\epsilon}$. Let
$\Lambda_{1}=\left\{ \lambda\in\Lambda_{0}:\,n_{\lambda}\in\mathcal{N}_{1}\right\} $.
The statement of the lemma then follows from the weak interlacing
of $\Lambda_{0}$ with $\mathcal{N}$.\end{proof}
\begin{lem}
\label{lem:BigSin}There exists a density one subsequence $\Lambda_{2}$
in $\Lambda_{0}$, such that for all $\lambda\in\Lambda_{2}$ and
$\xi\in\mathbb{Z}^{2}$ such that $\left|\xi\right|^{2}=n_{\lambda}$,
we have $\max\limits _{j=1,2}\left|\sin\left(\xi_{j}\alpha_{j}\right)\right|\gg\lambda^{-\epsilon}.$
\end{lem}
\begin{proof}
Denote the distance to the nearest integer by
\[
\left\Vert t\right\Vert =\min\limits _{n\in\mathbb{Z}}\left|t-n\right|.
\]
Since $\left|\sin\left(\xi_{j}\alpha_{j}\right)\right|\asymp\left\Vert \xi_{j}\alpha_{j}/\pi\right\Vert $,
it is enough to find a density one subsequence along which $\max\limits _{j=1,2}\left\Vert \xi_{j}\alpha_{j}/\pi\right\Vert \gg\lambda^{-\epsilon}$.

Let $\kappa$ be the type of $x_{0}/\pi$. Let 
\[
A=\left\{ \xi\in\mathbb{Z}^{2}:\,\left|\xi\right|^{2}\le X,\,\max\limits _{j=1,2}\left\Vert \xi_{j}\alpha_{j}/\pi\right\Vert \le X^{-\epsilon}\right\} .
\]
Then by writing $ \xi_1 = n, \xi_2 = n + h $, we have
\[
\#A\le\sum_{|h| \le 2X^{1/2}}\#A_{h}
\]
where
\[
A_{h}=\left\{ n\in\mathbb{Z}:\left|n\right|\le X^{1/2},\,\left\Vert n\alpha_{1}/\pi\right\Vert \le X^{-\epsilon},\,\left\Vert \left(n+h\right)\alpha_{2}/\pi\right\Vert \le X^{-\epsilon}\right\} .
\]
Fix $|h| \le 2X^{1/2},$ and divide the interval $\left[-X^{1/2},X^{1/2}\right]$
into subintervals of length $X^{\epsilon/\left(\kappa-1\right)},$
so the number of such intervals is $\asymp$ $X^{1/2-\epsilon/\left(k-1\right)}.$
For any $n\ne m$ which lie in one of these intervals, the distance
between the points $\left(\left\Vert n\alpha_{1}/\pi\right\Vert ,\left\Vert \left(n+h\right)\alpha_{2}/\pi\right\Vert \right)$
and $\left(\left\Vert m\alpha_{1}/\pi\right\Vert ,\left\Vert \left(m+h\right)\alpha_{2}/\pi\right\Vert \right)$
is bounded from below by
\[
\max_{j=1,2}\left\Vert \left(n-m\right)\alpha_{j}/\pi\right\Vert \gg\frac{1}{\left(n-m\right)^{\kappa-1}}\gg X^{-\epsilon},
\]
so the number of points in each of the intervals belonging to $A_{h}$
is bounded. Therefore $\#A_{h}\ll X^{1/2-\epsilon/\left(k-1\right)},$
and $\#A\ll X^{1-\epsilon/\left(k-1\right)}.$ Moreover, it follows
that

\begin{gather*}
\#\left\{ \xi\in\mathbb{Z}^{2}:\,\left|\xi\right|^{2}\le X,\,\max\limits _{j=1,2}\left\Vert \xi_{j}\alpha_{j}/\pi\right\Vert \le\left(1+\left|\xi\right|^{2}\right)^{-\epsilon}\right\} \ll X^{1-\epsilon/2\left(k-1\right)},
\end{gather*}
and thus the set 
\[
B=\left\{ n\in\mathcal{N}:\,n\le X,\,\exists\xi\in\mathbb{Z}^{2}.\,\left|\xi\right|^{2}=n\,\mbox{ and}\,\max\limits _{j=1,2}\left\Vert \xi_{j}\alpha_{j}/\pi\right\Vert \le\left(1+n\right)^{-\epsilon}\right\} 
\]
satisfies $\#B\ll X^{1-\epsilon/2\left(k-1\right)}.$ On the other
hand, since $r_{2}\left(n\right)\ll n^{\eta}$ for all $\eta>0$,
we have
\[
\#\left\{ n\in\mathcal{N}:\,n\le X\right\} \gg X^{1-\eta}
\]
(in fact, by Landau's theorem \cite{Landau} we have
$ \#\left\{ n\in\mathcal{N}:\,n\le X\right\} \sim C \frac{x}{\sqrt{\log x}} $).
Thus $\mathcal{N}\setminus B$ is a density one set in $\mathcal{N}$.
Let $\Lambda_{2}=\left\{ \lambda\in\Lambda_{0}:\,n_{\lambda}\in\mathcal{N}\setminus B\right\} $.
The statement of the lemma again follows from the weak interlacing
of $\Lambda_{0}$ with $\mathcal{N}$.
\end{proof}
Finally, we define $\Lambda'=\Lambda_{1}\cap\Lambda_{2}$ which is
a density one set in $\Lambda_{0}$ by Lemmas \ref{lem:NandLambdaClose}
and \ref{lem:BigSin}.

\subsection{Lower Bound for the  \texorpdfstring{$L^{2}$}{L2}-norm of \texorpdfstring{$G_\lambda$}{G\textlambda}}

For $\lambda\in\Lambda'$, let

\[
G_{\lambda}\left(x\right)=d_{1}G_{\lambda}\left(x,x_{1}\right)+d_{2}G_{\lambda}\left(x,x_{2}\right),
\]
normalized such that $\left|d_{1}\right|^{2}+\left|d_{2}\right|^{2}=1$.
Assume without loss of generality that $\left|d_{2}\right|^{2}\ge1/2$.

We now give a lower bound for the $L^{2}$-norm of $G_{\lambda}$
along $\lambda\in\Lambda'$:
\begin{lem}
For all $\lambda\in\Lambda'$, we have $\left\Vert G_{\lambda}\right\Vert _{2}^{2}\gg\lambda^{-4\epsilon}.$\end{lem}
\begin{proof}
Let $\lambda\in\Lambda'$. We have
\begin{alignat*}{1}
\left\Vert G_{\lambda}\right\Vert _{2}^{2} & =\frac{1}{16\pi^{4}}\sum_{\xi\in\mathbb{Z}^{2}}\frac{\left|d_{1}e^{i\left\langle \xi,x_{1}\right\rangle }+d_{2}e^{i\left\langle \xi,x_{2}\right\rangle }\right|^{2}}{\left(\left|\xi\right|^{2}-\lambda\right)^{2}}\\
 & =\frac{1}{16\pi^{4}}\left|d_{2}\right|^{2}\sum_{\xi\in\mathbb{Z}^{2}}\frac{\left|d_{1}/d_{2}+e^{i\left\langle \xi,x_{0}\right\rangle }\right|^{2}}{\left(\left|\xi\right|^{2}-\lambda\right)^{2}}.
\end{alignat*}
Choose $\xi=\left(\xi_{1},\xi_{2}\right)$ such that $\left|\xi\right|^{2}=n_{\lambda}$.
From Lemmas \ref{lem:NandLambdaClose} and \ref{lem:BigSin}, we have
$\max\limits _{j=1,2}\left|\sin\left(\xi_{j}\alpha_{j}\right)\right|\gg\lambda^{-\epsilon}$
and $n_{\lambda}-\lambda\ll\lambda^{\epsilon}.$ We can assume that
$\left| \sin\left(\xi_{1}\alpha_{1}\right) \right| \gg\lambda^{-\epsilon}$ (so in
particular $\xi_{1}\ne0$). Thus
\begin{alignat*}{1}
\left\Vert G_{\lambda}\right\Vert _{2}^{2} & \gg\sum_{\xi\in\mathbb{Z}^{2}}\frac{\left|d_{1}/d_{2}+e^{i\left\langle \xi,x_{0}\right\rangle }\right|^{2}}{\left(\left|\xi\right|^{2}-\lambda\right)^{2}}\gg\\
 & \gg\lambda^{-2\epsilon}\left(\left|d_{1}/d_{2}+e^{i\left\langle \left(\xi_{1},\xi_{2}\right),x_{0}\right\rangle }\right|^{2}+\left|d_{1}/d_{2}+e^{i\left\langle \left(-\xi_{1},\xi_{2}\right),x_{0}\right\rangle }\right|^{2}\right)\\
 & \gg\lambda^{-2\epsilon}\left|e^{i\left\langle \left(\xi_{1},\xi_{2}\right),x_{0}\right\rangle }-e^{i\left\langle \left(-\xi_{1},\xi_{2}\right),x_{0}\right\rangle }\right|^{2}\asymp\lambda^{-2\epsilon}\sin^{2}\left(\xi_{1}\alpha_{1}\right)\gg\lambda^{-4\epsilon}.
\end{alignat*}

\end{proof}

\subsection{Truncation}

Let $0<\delta<1/4$ and let $L=\lambda^{\delta}.$ We define $G_{\lambda,L}=d_{1}G_{\lambda,L}\left(x,x_{1}\right)+d_{2}G_{\lambda,L}\left(x,x_{2}\right)$
where
\[
G_{\lambda,L}\left(x,x_{j}\right)=-\frac{1}{4\pi^{2}}\sum_{\left|\left|\xi\right|^{2}-\lambda\right|\le L}\frac{e^{i\left\langle \xi,x-x_{j}\right\rangle }}{\left|\xi\right|^{2}-\lambda}
\]
is the truncated Green's function. Denote 
\[
g_{\lambda}=\frac{G_{\lambda}}{\left\Vert G_{\lambda}\right\Vert _{2}},\,g_{\lambda,L}=\frac{G_{\lambda,L}}{\left\Vert G_{\lambda,L}\right\Vert _{2}}.
\]

\begin{lem}
For all $\lambda\in\Lambda'$ we have $\left\Vert g_{\lambda}-g_{\lambda,L}\right\Vert _{2}^{2}\ll\lambda^{5\epsilon}/L$. \end{lem}
\begin{proof}
We have
\begin{alignat*}{1}
\left\Vert g_{\lambda}-g_{\lambda,L}\right\Vert _{2}^{2} & \le4\frac{\left\Vert G_{\lambda}-G_{\lambda,L}\right\Vert _{2}^{2}}{\left\Vert G_{\lambda}\right\Vert _{2}^{2}}\ll\lambda^{4\epsilon}\left\Vert G_{\lambda}-G_{\lambda,L}\right\Vert _{2}^{2}.
\end{alignat*}
But
\begin{alignat*}{1}
\left\Vert G_{\lambda}-G_{\lambda,L}\right\Vert _{2}^{2} & =\frac{1}{16\pi^{4}}\sum_{\left|\left|\xi\right|^{2}-\lambda\right|>L}\frac{\left|d_{1}e^{i\left\langle \xi,x_{1}\right\rangle }+d_{2}e^{i\left\langle \xi,x_{2}\right\rangle }\right|^{2}}{\left(\left|\xi\right|^{2}-\lambda\right)^{2}}\\
 & \ll\sum_{\left|\left|\xi\right|^{2}-\lambda\right|>L}\frac{1}{\left(\left|\xi\right|^{2}-\lambda\right)^{2}}\ll\sum_{\left|n-\lambda\right|>L}\frac{n^{\epsilon}}{\left(n-\lambda\right)^{2}}\ll\lambda^{\epsilon}/L.
\end{alignat*}

\end{proof}
For all $a\in C^{\infty}\left(\mathbb{T}^{2}\right)$, we have (see
\cite{RudnickUeberschaer})
\[
\left|\left\langle ag_{\lambda},g_{\lambda}\right\rangle -\left\langle ag_{\lambda,L},g_{\lambda,L}\right\rangle \right|\le2\left\Vert a\right\Vert _{\infty}\left\Vert g_{\lambda}-g_{\lambda,L}\right\Vert _{2}.
\]
Thus, 
\[
\left\langle ag_{\lambda},g_{\lambda}\right\rangle =\left\langle ag_{\lambda,L},g_{\lambda,L}\right\rangle +O\left(\lambda^{\left(5\epsilon-\delta\right)/2}\right).
\]
Taking $ \delta =1/4-\epsilon $, we see that in
order to prove Theorem \ref{thm:MainResultGeneral}, it is enough
to find a density one sequence $\Lambda_{\infty}$ in $\Lambda_{0}$
so that for all $a\in C^{\infty}\left(\mathbb{T}^{2}\right)$ and
all $M>0$,
\begin{equation}
\left\langle ag_{\lambda,L},g_{\lambda,L}\right\rangle =\frac{1}{4\pi^{2}}\int_{\mathbb{T}^{2}}a\left(x\right)\,\mbox{d}x+O\left(\lambda^{-M}\right)\label{eq:MainLimit}
\end{equation}
as $\lambda\to\infty$ along $\Lambda_{\infty}.$

\subsection{Proof of Theorem \ref{thm:MainResultGeneral}}

We first show that for every fixed $\zeta\in\mathbb{Z}^{2}\setminus\left\{ \left(0,0\right)\right\} $
\[
\left\langle e^{i\left\langle \zeta,x\right\rangle }g_{\lambda,L},g_{\lambda,L}\right\rangle =0
\]
along a density one subsequence.

Let
\[
S_{\zeta}=\left\{ \xi\in\mathbb{Z}^{2}:\,\left|\left\langle \xi,\zeta\right\rangle \right|\le2\left|\xi\right|^{2\delta}\right\} ,
\]
and let
\[
\Lambda_{\zeta}=\left\{ \lambda\in\Lambda':\,\forall\xi\in S_{\zeta}.\,\left|\left|\xi\right|^{2}-\lambda\right|>L\right\} .
\]

\begin{lem}
$\Lambda_{\zeta}$ is a density one set in $\Lambda_{0}$.\end{lem}
\begin{proof}
We follow the proof of Proposition 6.1 in \cite{RudnickUeberschaer}.
Write $\zeta=\left(p,q\right),\zeta^{\perp}=\left(-q,p\right)$. Then
every $\xi\in S_{\zeta}$ can be written as $\xi=u\frac{\zeta}{\left|\zeta\right|}+v\frac{\zeta^{\perp}}{\left|\zeta^{\perp}\right|}$,
and therefore the set of lattice points $\left\{ \xi\in S_{\zeta}:\,\left|\xi\right|^{2}\le X\right\} $
is contained in the rectangle 
\[
\left\{ u\frac{\zeta}{\left|\zeta\right|}+v\frac{\zeta^{\perp}}{\left|\zeta^{\perp}\right|}:\,u\le\frac{2X^{\delta}}{\left|\zeta\right|},\,v\le\sqrt{X}\right\} .
\]
Since the number of lattice points inside a rectangle is bounded (up
to a constant) by the area of the rectangle, we see that
\[
\#\left\{ \xi\in S_{\zeta}:\,\left|\xi\right|^{2}\le X\right\} \ll\frac{X^{1/2+\delta}}{\left|\zeta\right|}.
\]
Let $\mathcal{N_{\zeta}}\subseteq\mathcal{N}$ be the set of norms
$\left|\xi\right|^{2}$ in $S_{\zeta}$. Define a map $\phi:\Lambda'\setminus\Lambda_{\zeta}\to\mathcal{N}_{\zeta}$
which takes $\lambda\in\Lambda'\setminus\Lambda_{\zeta}$ to the closest
element $n\in\mathcal{N_{\zeta}}$ to $\lambda$ (if there are two
elements with the same distance take the smallest of them). For every
$\mathcal{N}_{\zeta}$ we have 
\[
\#\phi^{-1}\left(n\right)\le\#\left\{ \lambda\in\Lambda'\setminus\Lambda_{\zeta}:\,\exists\xi\in S_{\zeta}.\,\left|\xi\right|^{2}=n,\left|n-\lambda\right|\le L\right\} \ll n^{\delta}.
\]
Thus,
\begin{alignat*}{1}
\#\left\{ \lambda\in\Lambda'\setminus\Lambda_{\zeta}:\,\lambda\le X\right\}  & \le\sum_{\begin{subarray}{c}
n\in\mathcal{N}_{\zeta}\\
n\le2X
\end{subarray}}\#\phi^{-1}\left(n\right)\ll X^{\delta}\#\left\{ n\in\mathcal{N}_{\zeta}:n\le2X\right\} \\
 & \le X^{\delta}\#\left\{ \xi\in S_{\zeta}:\,\left|\xi\right|^{2}\le2X\right\} \ll\frac{X^{1/2+2\delta}}{\left|\zeta\right|},
\end{alignat*}
so $\Lambda_{\zeta}$ is a density one set in $\Lambda_{0}$ (since
$\delta<1/4).$ \end{proof}
\begin{lem}
\label{lem:xi_xizeta_distance}For all $\lambda\in\Lambda_{\zeta}$ such that $ \lambda^\delta \gg |\zeta|^2  $,
\[
\left|\left|\xi\right|^{2}-\lambda\right|\le L\Longrightarrow\left|\left|\xi+\zeta\right|^{2}-\lambda\right|>L.
\]
\end{lem}
\begin{proof}
For all $\lambda\in\Lambda_{\zeta}$ such that $ \lambda^\delta \gg |\zeta|^2  $, if $\left|\left|\xi\right|^{2}-\lambda\right|\le L$, then
$\xi\notin S_{\zeta},$ i.e., $\left|\left\langle \xi,\zeta\right\rangle \right|>2\left|\xi\right|^{2\delta},$
and therefore
\[
\left|\left|\xi+\zeta\right|^{2}-\lambda\right|\ge2\left|\left\langle \xi,\zeta\right\rangle \right|-\left|\left|\xi\right|^{2}-\lambda\right|-\left|\zeta\right|^{2}>L.
\]

\end{proof}
\emph{\phantom{}Proof of Theorem \ref{thm:MainResultGeneral}:} We
have
\begin{alignat*}{1}
\left\langle e^{i\left\langle \zeta,x\right\rangle }g_{\lambda,L},g_{\lambda,L}\right\rangle  & =\frac{\left\langle e^{i\left\langle \zeta,x\right\rangle }G_{\lambda,L},G_{\lambda,L}\right\rangle }{\left\Vert G_{\lambda,L}\right\Vert _{2}^{2}}.
\end{alignat*}
Denoting $c\left(\xi\right)=d_{1}e^{-i\left\langle \xi,x_{1}\right\rangle }+d_{2}e^{-i\left\langle \xi,x_{2}\right\rangle },$
note that
\[
\left\langle e^{i\left\langle \zeta,x\right\rangle }G_{\lambda,L},G_{\lambda,L}\right\rangle =\frac{16}{\pi^{4}}\sum_{\begin{subarray}{c}
\left|\left|\xi\right|^{2}-\lambda\right|\le L\\
\left|\left|\xi+\zeta\right|^{2}-\lambda\right|\le L
\end{subarray}}\frac{c\left(\xi\right)\overline{c\left(\xi+\zeta\right)}}{\left(\left|\xi\right|^{2}-\lambda\right)\left(\left|\xi+\zeta\right|^{2}-\lambda\right)}.
\]
However by Lemma \ref{lem:xi_xizeta_distance}, the last sum is empty
along $\lambda\in\Lambda_{\zeta}$ such that $ \lambda^\delta \gg |\zeta|^2  $, so along this sequence $\left\langle e^{i\left\langle \zeta,x\right\rangle }g_{\lambda,L},g_{\lambda,L}\right\rangle =0$.

We conclude (\ref{eq:MainLimit}) by an argument which can be found
in \cite{Ueberschaer3}. We expand $a$ into a Fourier series:
\[
a\left(x\right)=\sum_{\zeta\in\mathbb{Z}^{2}}\hat{a}\left(\zeta\right)e^{i\left\langle \zeta,x\right\rangle }.
\]
By the rapid decay of the Fourier coefficients $\hat{a}\left(\zeta\right),$
we see that for any $M>0$
\[
a\left(x\right)=\sum_{\left|\zeta\right|\le\lambda^{\epsilon}}\hat{a}\left(\zeta\right)e^{i\left\langle \zeta,x\right\rangle }+O\left(\lambda^{-M}\right).
\]
Define
\[
\Lambda_{\infty}=\left\{ \lambda\in\Lambda':\,\forall\left|\zeta\right|\le\lambda^{\epsilon}.\,\lambda\in\Lambda_{\zeta}\right\} .
\]
We have
\begin{alignat*}{1}
\#\left\{ \lambda\in\Lambda'\setminus\Lambda_{\infty}:\,\lambda\le X\right\}  & \le\sum_{\left|\zeta\right|\le X^{\epsilon}}\#\left\{ \lambda\in\Lambda'\setminus\Lambda_{\zeta}:\,\lambda\le X\right\} \\
 & \ll X^{1/2+2\delta}\sum_{0<\left|\zeta\right|\le X^{\epsilon}}\frac{1}{\left|\zeta\right|}\ll X^{1/2+2\delta+\epsilon},
\end{alignat*}
so $\Lambda_{\infty}$ is a density one set in $\Lambda_{0}$. Finally,
for any $\lambda\in\Lambda_{\infty}$
\begin{alignat*}{1}
\left\langle ag_{\lambda,L},g_{\lambda,L}\right\rangle  & =\hat{a}\left(0\right)+\sum_{0<\left|\zeta\right|\le\lambda^{\epsilon}}\hat{a}\left(\zeta\right)\left\langle e^{i\left\langle \zeta,x\right\rangle }g_{\lambda,L},g_{\lambda,L}\right\rangle +O\left(\lambda^{-M}\right)\\
 & =\hat{a}\left(0\right)+O\left(\lambda^{-M}\right),
\end{alignat*}
since for each $\left|\zeta\right|\le\lambda^{\epsilon},$ $\left\langle e^{i\left\langle \zeta,x\right\rangle }g_{\lambda,L},g_{\lambda,L}\right\rangle =0.$
Thus, the limit (\ref{eq:MainLimit}) holds along $\Lambda_{\infty},$
and Theorem \ref{thm:MainResultGeneral} follows.

\qed

\appendix

\section{\label{sec:Appendix-A}}

We study the eigenspaces of a self-adjoint extension $-\Delta_{U}$
corresponding to old eigenvalues of $-\Delta$. Our goal is to show
that their dimensions are equal to the dimensions of the eigenspaces
of $-\Delta$ minus $\mbox{rank}\left(I+U\right)$. We first prove
three auxiliary lemmas:
\begin{lem}
\label{lem:OldEigenspaceDim}Let $0\ne\lambda\in\sigma\left(-\Delta\right),$
and let $d$ be the dimension of the corresponding eigenspace $E_{\lambda}$.
Then the dimension of the subspace 
\[
\left\{ f\in E_{\lambda}:\,f\left(x_{1}\right)=f\left(x_{2}\right)=0\right\} 
\]
is equal to $d-2$.\end{lem}
\begin{proof}
Since $x_{0}/\pi$ is Diophantine, we can assume that $\alpha_{1}/\pi\notin\mathbb{Q}$.
Fix $\xi=\left(\xi_{1},\xi_{2}\right)$ such that $\xi_{1}\ne0$,
$\left|\xi\right|^{2}=\lambda$. The functions 
\[
\left\{ e^{i\left\langle x-x_{1},\eta\right\rangle }-e^{i\left\langle x-x_{1},\xi\right\rangle }\right\} _{\left|\eta\right|^{2}=\lambda,\,\eta\ne\xi}
\]
form a basis for the subspace $\left\{ f\in E_{\lambda}:\,f\left(x_{1}\right)=0\right\} .$
Choose $\eta=\left(-\xi_{1},\xi_{2}\right)$, and let g$\left(x\right)=e^{i\left\langle x-x_{1},\eta\right\rangle }-e^{i\left\langle x-x_{1},\xi\right\rangle }$.
Then 
\[
\left|g\left(x_{2}\right)\right|=2\left|\sin\left(\xi_{1}\alpha_{1}\right)\right|\ne0
\]
since $\alpha_{1}/\pi\notin\mathbb{Q},$ and therefore the functions
\[
\left\{ e^{i\left\langle x-x_{1},\zeta\right\rangle }-e^{i\left\langle x-x_{1},\xi\right\rangle }-\frac{g\left(x\right)}{g\left(x_{2}\right)}\left(e^{i\left\langle x_{0},\zeta\right\rangle }-e^{i\left\langle x_{0},\xi\right\rangle }\right)\right\} _{\left|\zeta\right|^{2}=\lambda,\,\zeta\ne\xi,\eta}
\]
form a basis for the subspace $\left\{ f\in E_{\lambda}:\,f\left(x_{1}\right)=f\left(x_{2}\right)=0\right\} $.
\end{proof}
Recall that if $\mbox{rank}\left(I+U\right)=1$, $v_{0}$ is defined
to be the (unique up to scalar) non-zero vector such that $\left(I+U^{*}\right)v_{0}=0.$
We have the following property for $\left\langle T^{*}v_{0},\mbox{Im}\mathbb{G}_{i}\left(x\right)\right\rangle $:
\begin{lem}
\label{lem:ImGiNotEigenfunction}Assume that $\mbox{rank}\left(I+U\right)=1$.
Then $\left\langle T^{*}v_{0},\mbox{Im}\mathbb{G}_{i}\left(x\right)\right\rangle $
is not an eigenfunction of $-\Delta$. \end{lem}
\begin{proof}
Denote $T^{*}v_{0}=\left(v_{1},v_{2}\right)$, so
\[
\left\langle T^{*}v_{0},\mbox{Im}\mathbb{G}_{i}\left(x\right)\right\rangle =-\frac{1}{4\pi^{2}}\sum_{\xi\in\mathbb{Z}^{2}}\left(v_{1}e^{-i\left\langle \xi,x_{1}\right\rangle }+v_{2}e^{-i\left\langle \xi,x_{2}\right\rangle }\right)\frac{e^{i\left\langle \xi,x\right\rangle }}{\left|\xi\right|^{4}+1}.
\]
Assume that $\left\langle T^{*}v_{0},\mbox{Im}\mathbb{G}_{i}\left(x\right)\right\rangle $
is an eigenfunction of $-\Delta$ with an eigenvalue $\lambda.$ Then
for all $\xi$ such that $\left|\xi\right|^{2}=m\ne\lambda$ we have
\[
v_{1}e^{-i\left\langle \xi,x_{1}\right\rangle }+v_{2}e^{-i\left\langle \xi,x_{2}\right\rangle }=0.
\]
We can again assume that $\alpha_{1}/\pi\notin\mathbb{Q}.$ Choosing
any $\xi=\left(\xi_{1},\xi_{2}\right)$ with $\xi_{1}\ne0$ and $\left|\xi\right|^{2}=m\ne\lambda$,
we get in particular that
\[
\det\left(\begin{array}{cc}
1 & e^{-i\left\langle \left(\xi_{1},\xi_{2}\right),x_{0}\right\rangle }\\
1 & e^{-i\left\langle \left(-\xi_{1},\xi_{2}\right),x_{0}\right\rangle }
\end{array}\right)=0,
\]
however since $\alpha_{1}/\pi\notin\mathbb{Q},$ we have
\[
\left|\det\left(\begin{array}{cc}
1 & e^{-i\left\langle \left(\xi_{1},\xi_{2}\right),x_{0}\right\rangle }\\
1 & e^{-i\left\langle \left(-\xi_{1},\xi_{2}\right),x_{0}\right\rangle }
\end{array}\right)\right|=2\left|\sin\left(\xi_{1}\alpha_{1}\right)\right|\ne0
\]
a contradiction.\end{proof}
\begin{lem}
\label{lem:OldEigenspacesGi}Assume that $\mbox{rank}\left(I+U\right)=1$.
Let $0\ne\lambda\in\sigma\left(-\Delta\right),$ and let $d$ be the
dimension of the corresponding eigenspace $E_{\lambda}$. Then the
dimension of the subspace 
\[
\left\{ g\in E_{\lambda}:\,g\left(x\right)=f\left(x\right)+c\left\langle T^{*}v_{0},\mbox{Im}\mathbb{G}_{i}\left(x\right)\right\rangle ,f\left(x_{1}\right)=f\left(x_{2}\right)=0,c\in\mathbb{C}\right\} 
\]
is equal to $d-1.$\end{lem}
\begin{proof}
From Lemma \ref{lem:ImGiNotEigenfunction}, $\left\langle T^{*}v_{0},\mbox{Im}\mathbb{G}_{i}\left(x\right)\right\rangle \notin E_{\lambda}$,
and therefore 
\[
\dim\left(E_{\lambda}+\left\langle T^{*}v_{0},\mbox{Im}\mathbb{G}_{i}\left(x\right)\right\rangle \right)=d+1.
\]
The proof of the statement of the lemma now follows similarly to the
proof of Lemma \ref{lem:OldEigenspaceDim}.\end{proof}
\begin{prop}
\label{prop:OldEigenspaces}Let $0\ne\lambda\in\sigma\left(-\Delta\right),$
and assume that $g$ is an eigenfunction of $-\Delta_{U}$ with an
eigenvalue $\lambda$. Then $g$ is an eigenvalue of $-\Delta$. \end{prop}
\begin{proof}
Assume otherwise, so there exist $v\notin\mbox{Ker}\left(I+U^{*}\right)$
and $f\in H^{2}\left(\mathbb{T}^{2}\right)$ with $f\left(x_{1}\right)=f\left(x_{2}\right)=0$
such that 
\[
g\left(x\right)=f\left(x\right)+\left\langle T^{*}v,\mathbb{G}_{i}\left(x\right)\right\rangle +\left\langle \left(UT\right)^{*}v,\mathbb{G}_{-i}\left(x\right)\right\rangle 
\]
where $g$ is an eigenvalue of $-\Delta_{U}$, and hence of the adjoint
operator $-\Delta_{0}^{*}.$ Thus, as in (\ref{eq:EigenfunctionEq})
\[
0=\left(\Delta_{0}^{*}+\lambda\right)g=\left(\Delta+\lambda\right)f+\left(-i+\lambda\right)\left\langle T^{*}v,\mathbb{G}_{i}\right\rangle +\left(i+\lambda\right)\left\langle \left(UT\right)^{*}v,\mathbb{G}_{-i}\right\rangle .
\]
Assume that $\alpha_{1}/\pi\notin\mathbb{Q},$ and let $\mbox{\ensuremath{\xi}=\ensuremath{\left(\xi_{1},\xi_{2}\right)} }$
such that $\xi_{1}\ne0$ and $\left|\xi\right|^{2}=\lambda$. Evaluating
the Fourier coefficient at $\left(\pm\xi_{1},\xi_{2}\right)$, we
see that
\[
\left\langle T^{*}\left(I+U^{*}\right)v,\left(1,e^{-i\left\langle \left(\xi_{1},\xi_{2}\right),x_{0}\right\rangle }\right)\right\rangle =0
\]
and
\[
\left\langle T^{*}\left(I+U^{*}\right)v,\left(1,e^{-i\left\langle \left(-\xi_{1},\xi_{2}\right),x_{0}\right\rangle }\right)\right\rangle =0.
\]
Since $v\notin\mbox{Ker}\left(I+U^{*}\right)$ it implies that 
\[
\left|\det\left(\begin{array}{cc}
1 & e^{-i\left\langle \left(\xi_{1},\xi_{2}\right),x_{0}\right\rangle }\\
1 & e^{-i\left\langle \left(-\xi_{1},\xi_{2}\right),x_{0}\right\rangle }
\end{array}\right)\right|=0,
\]
a contradiction.\end{proof}
\begin{cor}
Let $0\ne\lambda\in\sigma\left(-\Delta\right).$ Let $d$ be the dimension
of the eigenspace of $-\Delta$ corresponding to $\lambda.$ Then
the dimension of the eigenspace of $-\Delta_{U}$ corresponding to
$\lambda$ is equal to $d-\mbox{rank}\left(I+U\right).$ \end{cor}
\begin{proof}
If $\mbox{rank}\left(I+U\right)=2$ and $g$ is an eigenvalue of $-\Delta_{U}$
with $0\ne\lambda\in\sigma\left(-\Delta\right)$, then from the proof
of Lemma \ref{prop:OldEigenspaces}, we have
\[
g\left(x\right)=f\left(x\right)+\left\langle T^{*}v,\mathbb{G}_{i}\left(x\right)\right\rangle +\left\langle \left(UT\right)^{*}v,\mathbb{G}_{-i}\left(x\right)\right\rangle 
\]
with $f\left(x_{1}\right)=f\left(x_{2}\right)=0$ and $\left(I+U^{*}\right)v=0$,
and therefore $v=0$, so $g\in H^{2}\left(\mathbb{T}^{2}\right)$
with $g\left(x_{1}\right)=g\left(x_{2}\right)=0$, and the statement
follows from Lemma \ref{lem:OldEigenspaceDim} .

If $\mbox{rank}\left(I+U\right)=1$, then there exist $0\ne v_{0}\in\mathbb{C}^{2}$,
so that
\[
g\left(x\right)=f\left(x\right)+c\left\langle T^{*}v_{0},\mbox{Im}\mathbb{G}_{i}\left(x\right)\right\rangle 
\]
with $f\left(x_{1}\right)=f\left(x_{2}\right)=0$, $c\in\mathbb{C}$.
Thus, in this case the statement follows from Lemma \ref{lem:OldEigenspacesGi}
.\end{proof}

\end{document}